\documentclass[pre,floatfix,superscriptaddress,nofootinbib, longbibliography,onecolumn,final,notitlepage]{revtex4-1}
\usepackage{graphicx}
\usepackage[utf8]{inputenc}
\usepackage[english]{babel}
\usepackage{microtype}
\usepackage{booktabs} 
\usepackage{amsmath}
\usepackage{amsfonts}
\usepackage{amsthm}
\usepackage{amssymb}
\usepackage{dsfont}
\usepackage{algorithm}
\usepackage{algpseudocode}
\usepackage{bm}

\usepackage[justification=raggedright,singlelinecheck=false]{caption}

\usepackage{tabularx}
\usepackage{subcaption}
\usepackage{array}
\usepackage[normalem]{ulem}

\newtheorem{theorem}{Theorem}

\usepackage[colorlinks, breaklinks,hyperfootnotes = false,citecolor = blue]{hyperref}
\usepackage{natbib}
\usepackage[capitalise]{cleveref}

\begin{document}

\title{Inferring signed social networks from contact patterns}
\author{D\'avid Ferenczi}
\email{david.ferenczi@maastrichtuniversity.nl}
\affiliation{Department of Data Analytics and Digitalisation, School of Business and Economics, Maastricht University, 6211 LM Maastricht, Netherlands}
\author{Jean-Gabriel Young}
\email{jean-gabriel.young@uvm.edu}
\affiliation{Department of Mathematics and Statistics, University of Vermont, Burlington, VT 05405, USA}
\affiliation{Vermont Complex Systems Institute, University of Vermont, Burlington, VT 05405, USA}
\author{Leto Peel}
\email{l.peel@maastrichtuniversity.nl}
\affiliation{Department of Data Analytics and Digitalisation, School of Business and Economics, Maastricht University, 6211 LM Maastricht, Netherlands}

\begin{abstract}
    Social networks are typically inferred from indirect observations, such as proximity data; yet, most methods cannot distinguish between absent relationships and actual negative ties, as both can result in few or no interactions. 
    We address the challenge of inferring signed networks from contact patterns while accounting for whether a lack of interactions reflects a lack of opportunity as opposed to active avoidance.
    We develop a Bayesian framework with MCMC inference that models interaction groups to separate chance from choice when no interactions are observed. 
    Validation on synthetic data demonstrates superior performance compared to natural baselines, particularly in detecting negative edges. 
    We apply our method to French high school contact data to reveal a structure consistent with friendship surveys and demonstrate the model's adequacy through posterior predictive checks.
\end{abstract}

\maketitle

\section{Introduction}

Networks are often used to represent complex social systems, where nodes represent individuals and edges represent relationships between them. 
Signed networks can provide a richer representation of social systems than their unsigned counterparts by allowing for the distinction between positive and negative relationships according to the signs that annotate the edges.
This distinction is crucial for understanding social dynamics, as positive and negative relationships can have different effects on individual behavior and group dynamics~\cite{harrigan2020negative}.
Signed networks are used in social science to study phenomena such as social balance~\cite{harary1953notion}, conflict and cooperation~\cite{hiller2017friends}, and the spread of information and influence~\cite{liu2019influence, he2019information}. 
More recently, signed networks have been used in the analysis of online social media platforms~\cite{leskovec2010signed, tang2016survey}, including analyzing trust and reputation in online platforms~\cite{guha2004propagation}, identifying polarized communities~\cite{polarized_communities} and modelling recommendation systems~\cite{tang2016recommendations}. Furthermore, dynamical models applied to signed networks have proven valuable in examining how structural balance and factions evolve over time when interactions are driven by perceived sentiments~\cite{shang2020}.	

Networks are not, in general, directly observable---they are a mathematical abstraction that must be connected to observations by models~\cite{butts2009revisiting, jacobs2021measurement}.
We must therefore reconstruct networks from data based on observations or measurements that are typically indirect and/or noisy~\cite{peel2022statistical}.
Social networks are no exception as typically the social relationships we wish to study are not directly observable.
Instead, we rely on reconstructing these relationships from data that include surveys, proximity measurements, or communication logs~\cite{Mastrandrea_Fournet_Barrat_2015}---a process complicated by measurement errors and the inherent ambiguity of infrequent interactions~\cite{redhead2023reliable,de2023latent}. 
In this work, we address the challenge of reconstructing positive and negative social relationships from contact interactions based on frequency of physical proximity between pairs of individuals over time. 
This type of contact data is often collected in studies of human~\cite{cattuto2010dynamics, Mastrandrea_Fournet_Barrat_2015} and animal~\cite{whitehead2008analyzing} social behavior, and typically sourced by manually observing subjects or using sensors that automatically record proximity or interactions~\cite{eagle2006reality, cattuto2010dynamics}.

We approach this reconstruction problem with the assumption that the data we observe are \emph{generated} by an unobserved signed social network.
If two individuals are close friends, for example, we might expect them to interact more frequently and often find them in physical proximity.
This connection between a latent quantity (the social network) and observations (interaction frequencies) enables us to make inferences about what the network might look like without ever directly observing it.

Various statistical models coupled with Bayesian inference have previously been used to reconstruct networks from noisy and/or incomplete data~\cite{young2021bayesian,peixoto2018reconstructing,guimera2009missing} and from indirect observations~\cite{braunstein2019network,peixoto2019network}. However, these methods do not provide the means to reconstruct signed networks. A number of methods have been proposed to reconstruct signed social networks. 
Some of these methods are based on reconstructing edges independently, framed as a regression model~\cite{10.1145/1772690.1772756}, while others reconstruct networks holistically, incorporating approaches such as bipartite projection and backbone extraction~\cite{NEAL202280}, hypergeometric random graphs~\cite{Andres2023}, and temporal network-based approaches~\cite{gelardi2019detecting}. 
However, these methods for signed network reconstruction all struggle to distinguish between the absence of a relationship and the presence of an actual negative tie, as both can manifest as infrequent interactions in the data.
This ambiguity presents a critical challenge in signed network reconstruction. 
Consequently, existing methods tend to over or underrepresent negative edges, either by misclassifying absent relationships as negative ties or by failing to detect genuine negative relationships due to a lack of observed interactions. 

We address these challenges by developing a data generating process that explicitly accounts for interaction opportunities through the concept of interaction groups.
This approach enables us to distinguish whether a lack of interaction results from the absence of opportunity (chance) or active avoidance (choice), thereby providing a principled framework for differentiating whether a lack of interactions occurs by chance or by choice.
We present a Markov chain Monte Carlo inference algorithm to reconstruct signed networks and provide uncertainty quantification for our estimates.
We validate our method on synthetic data, demonstrating its superior performance, particularly in detecting negative edges, compared to natural baselines.
Finally, we apply our method to real-world contact data collected in a French high school~\cite{Mastrandrea_Fournet_Barrat_2015}, revealing a signed network structure consistent with friendship surveys and demonstrating our method's adequacy through posterior predictive checks.

\section{Generating interaction data from signed networks}
\label{sec:model}
Our approach to signed network reconstruction relies on making a connection between observed pairwise frequencies of interactions and a latent signed network of relationships by defining a generative process. 
The latent signed network $\bm A$ is a symmetric matrix, in which each entry $A_{ij}$ can be $-1$, $+1$, or $0$ to represent negative, positive, or absent relationships respectively.
We assume that a pair of individuals will only interact with each other according to their latent relationship $A_{ij}$ if they have the opportunity to do so, i.e., if they are in the same location at a given time.
This information is encoded in the partition vector $\bm{g}\in \mathbb{N}^n$ that assigns each of the $n$ nodes to a group for the duration of a measurement period.
Figure \ref{fig:interaction_mod} gives a conceptual overview of how the probability of interaction depends on the latent network and group assignment.
If two individuals $i$ and $j$ are in the same group, i.e., $g_i=g_j$, then they have the opportunity to freely interact.
The frequency of interactions between $i$ and $j$ depends on the 
probability of interaction $p_{A_{ij}} \in \bm{p} = \{p_{-1},p_0,p_1\}$, which corresponds to the interaction probability that is conditioned on the latent signed relationship $A_{ij}$.
However, if $g_i\neq g_j$, then individuals $i$ and $j$ do not have the opportunity to interact freely and instead interact with probability $q$ that represents the rate of random chance encounters. 
The frequency of observed interactions is recorded in the observation matrix $\bm{X}$, where each element of the observation matrix $X_{ij}$ records the number of times two individuals $i$ and $j$ interact with one another.

\begin{figure*}
	\centering
	\includegraphics[width=0.8\linewidth]{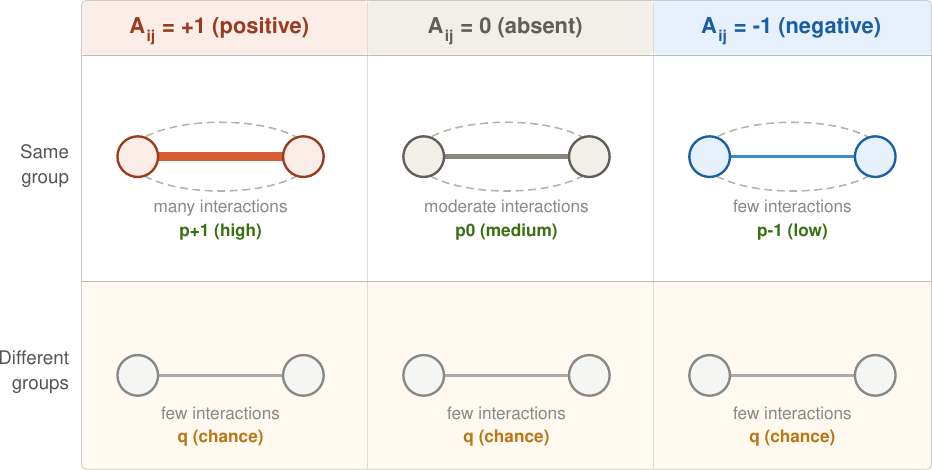}
	\caption{Illustration of how interaction probabilities depend on the latent social network and group membership. The model separates chance, governed by the inter-group interaction rate $q$, from choice, governed by the intra-group interaction rates $p_{-1}, p_0, p_{1}$.}
	\label{fig:interaction_mod}
\end{figure*}

Figure~\ref{fig:figure_example_data_gen} shows an overview of the model and inference task.
We can see, for example, that positive relationships result in more interactions and that individuals interact more frequently with other individuals only if given the opportunity. 
We also observe that the final result, our inference target, is a probabilistic description of the latent interactions $\bm{A}$, which is summarized here by marginal probabilities for each edge.
Even if $\bm{A}$ is never observed directly, Bayesian inference allows us to reconstruct the network's structure well.

The idea of using network partitions to modulate probability of interaction resembles network reconstruction approaches~\cite{guimera2009missing,peixoto2018reconstructing} based on the stochastic block model~\cite{holland1983stochastic, nowicki2001estimation}, where edges are generated based on group memberships. 
An important difference between our model and these previous approaches is that in our model the partition represents interaction opportunities and therefore influences the structure in the data, rather than directly influence structure of the network itself.

\begin{figure}[t]
    \includegraphics[width=0.8\columnwidth]{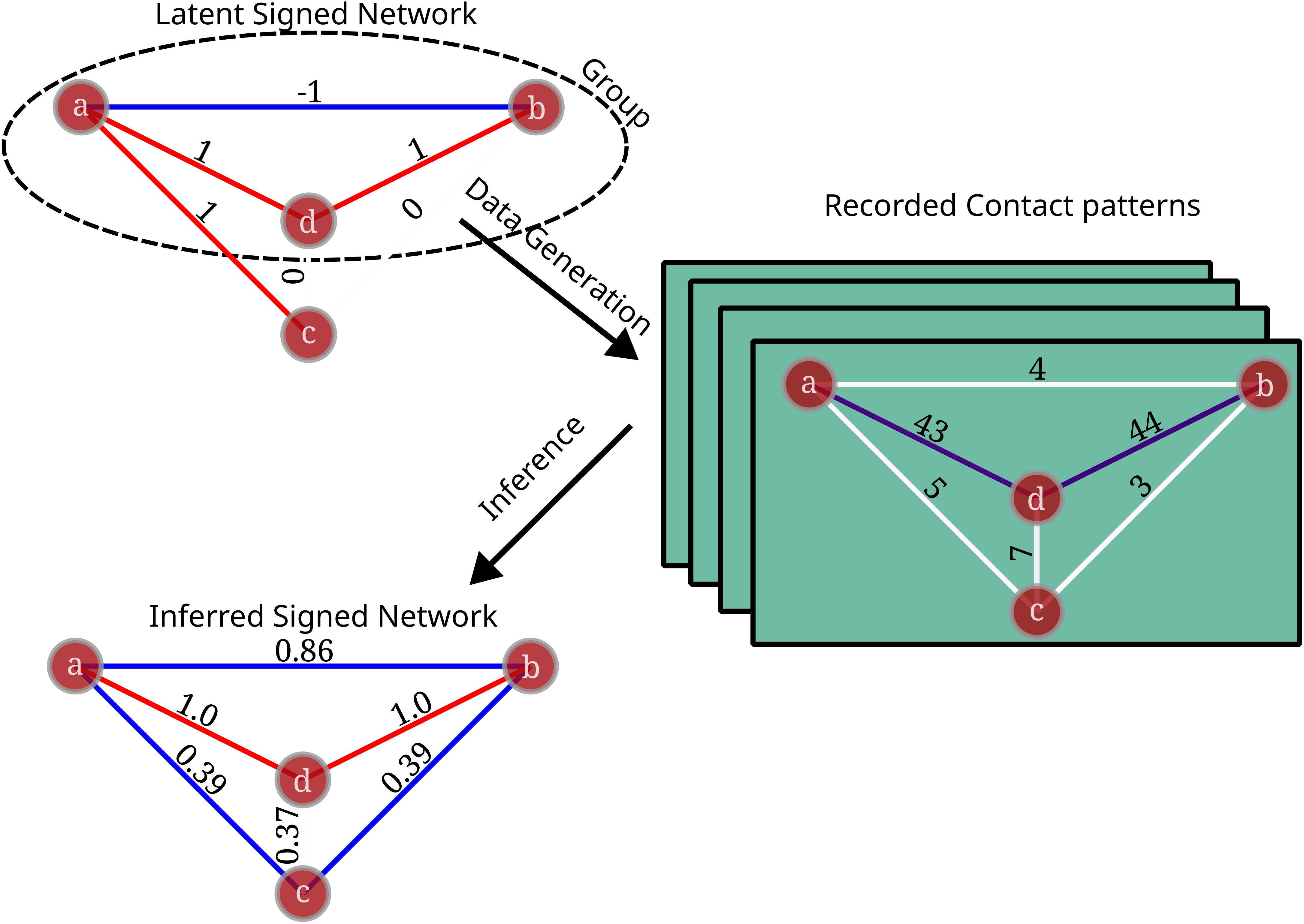}
    \caption{
        Data generation and network reconstruction process. 
        (\textbf{Top}) The latent signed network is the target of inference. 
        (\textbf{Middle}) We observe recorded contact patterns created with the binomial model of Eq.~\eqref{eq:obs_desc}.
        In this example, individuals $a$, $b$, and $d$ end up in close proximity (i.e., in the same group) and can interact with high probability, while individual $c$ is by itself.
        This leads to, for example, 43 recorded interactions between individuals $a$ and $d$, a comparatively high number since they have a positive relationship, and four interactions between individuals $a$ and $b$, since they have a negative relationship.
        Individual $c$ has few interactions with others, as it is not in the same group.
        (\textbf{Bottom}) Our inference algorithm ascribes a probability to every interaction from the recorded contact patterns alone.
        For instance, we find that the interaction between $a$ and $d$ is very likely positive (probability 1.0), while $a$ and $b$ very likely have a negative relationship (probability 0.86).
        We are less certain about the interaction between $c$ and the other individuals because the model treats observations between groups as noisy.
    }
    \label{fig:figure_example_data_gen}
\end{figure}

Our model has two components: (i) a component that generates a latent signed network $\bm A$, and 
(ii) a component that generates the observation matrix $\bm X$ conditioned on the signed network $\bm A$, the partition $\bm{g}$, and the interaction rates $\bm{p},q$. 
Combining these two components with priors for the partition $\mathbb P(\bm g)$ and the parameters $f(\bm p,q)$ completes our Bayesian model and allows us to calculate the posterior density over signed networks and parameters,
\begin{equation}
    \label{eq:bayes}
    f(\bm{A},\bm g, \bm p,q\mid\bm{X}) 
    \propto \mathbb P(\bm{X}\mid\bm{A}, \bm g, \bm p)\, \mathbb P(\bm g)\, \mathbb P(\bm{A})\, f(\bm p,q) \enspace.
\end{equation}

The central task at hand is then to estimate the network $\bm{A}$.
The posterior distribution also contains information about the partition $\bm g$ and the model parameters~\cite{young2021bayesian}, however since these are not our focus, we can marginalize them out,
\begin{equation}
    \label{eq:bayes:marg}
    \mathbb P(\bm{A}\mid\bm{X}) \propto \sum_{\bm g} \int_{\bm p,q }  f(\bm{A},\bm g, \bm p,q\mid\bm{X}) \enspace.
\end{equation}

This distribution tells us which signed networks are likely given an observation matrix $\bm{X}$. 
To obtain a point estimate of the reconstructed network, we can compute the mean of the posterior marginal: 
\begin{align}
    \label{eq:marginal_mean}
    \hat{A}_{ij} =  (-1) \cdot P(A_{ij} = -1 \mid X) 
     + 0 \cdot P(A_{ij} = 0 \mid X) 
    + (+1) \cdot P(A_{ij} = +1 \mid X) \enspace.
\end{align}

\subsection{Signed network}
\label{subsec:partition_model}
We start by specifying the details of generating latent signed networks. 
We represent each element $A_{ij}$ as being generated from a simple categorical distribution,
\begin{equation}
    A_{ij}\sim\mathrm{Categorical}(\rho) \enspace,
\end{equation}
where the possible outcomes are $A_{ij}\in\{-1,0,1\}$, and where $\rho=[\rho_{-1},\rho_0,\rho_1]$ is a vector of prior probabilities of each edge type.
This choice assumes that all edges are independent \textit{a priori}, leading to our prior distribution over networks,
\begin{equation}
     \mathbb P(\bm{A}\mid\rho) =  \prod_{1\leq i<j\leq n} \rho_{A_{ij}} \enspace.
\end{equation} 

It would be possible to incorporate a hierarchical prior for $\rho$ in our model, allowing us to infer these values from the data, but at a significantly increased computational cost.
Instead, we fixed the values of $\rho=(1/3, 1/3, 1/3)$ as we found that this choice worked well in practice and setting different values did not have a significant impact on the result; see Sec.~\ref{sec:prior_experiments}.
This choice formally reduces the network prior to a constant,
\begin{align}
     \mathbb P(\bm{A}) =  \left(\frac{1}{3}\right)^{\binom{n}{2}} \enspace.
\end{align}


\subsection{Observation matrix} 
\label{sec:priors}

A core assumption of our model is that individuals interact with a frequency according to their latent relationships if given the opportunity.
The opportunity to interact is encoded in the partition $\bm{g}$, for which we choose a uniform prior over all possible partitions, 
\begin{equation}
    \mathbb P(\bm {g}) =1/B_n \enspace,
\end{equation}
where $n$ is the number of nodes in the network and $B_n$ is the $n$-th Bell number. 

When a pair of individuals $i$ and $j$ are in the same group, i.e., $g_i=g_j$, they have the opportunity to interact according to their relationship $\bm A_{ij}$ with probability $\bm{p}=\{p_{-1}, p_0, p_1\}$.
When they are in different groups, i.e., $g_i\neq g_j$, they interact with probability $q$ that represents random chance encounters.
We model these opportunity groups $\bm{g}$ as stable for a given observation window, though of course, in reality these groups may change over time. We account for this variability by using short observation windows, such that social groups do not reorganize significantly.

To encode the idea that positive relationships lead to more interactions, we impose the constraint $0< p_{-1}< p_0<p_1<1$.
We make no further assumptions beyond this constraint and model these variables as uniform random variables.
The joint density of $p$ and $q$ is thus a constant, since we can think of the distribution over $p$ as a distribution over a 2-simplex (of volume $1/3!$), and $f(q)=1$. This leads to
\begin{equation}
    f(\bm p,q) = f(\bm p)f(q)= 6 \cdot \mathds 1_{p_{-1}\leq p_0\leq p_1} (\bm p) \enspace,
\end{equation}
where $\mathds 1_A(x)$ is an indicator function, equal to $1$ if $x\subset A$ and $0$ otherwise.

With these definitions in place, we model the random variable $X_{ij}$, corresponding to the observed number of interactions between individuals $i$ and $j$ throughout a data collection period of length $t$ as a binomial random variable.
We distinguish two cases: a pair of individuals are either in the same group or not,
\begin{align}
    \label{eq:obs_desc}
    \bm X_{ij}=
    \left\{
        \begin{array}{ll}
            \mathrm{Binomial}(t,p_{A_{ij}})&\text{if } g_i= g_j \enspace,\\
            \mathrm{Binomial}(t,q)&\text{if } g_i\neq g_j \enspace,
        \end{array}
    \right.
\end{align}
where the index $A_{ij}$ indicates the relevant sign for the interaction.

We assume these pairwise observations are conditionally independent given the latent network and the parameters. 
This conditional independence assumption allows us to write the likelihood---the probability of the observations---as
\begin{align}
    \label{eq:likelihood}
    \mathbb{P}(\bm{X}\mid \bm{A}, \bm{g}, \bm{p},q) 
    = & \prod_{1\leq i<j \leq n} \left[{t \choose X_{ij}} p_{A_{ij}}^{X_{ij}}\cdot (1- p_{A_{ij}})^{t-X_{ij}}\right]^{\delta_{ g_i g_j}}\times\left[
    {t \choose X_{ij}} q^{X_{ij}}(1-q)^{t-X_{ij}}\right]^{1-\delta_{g_i g_j}} \enspace,
\end{align}
where $n \choose x$ is the binomial coefficient and $\delta_{ij}$ is the Kronecker delta.

\subsection{Posterior distribution}
Substituting all of the above equations into the right-hand side of Eq.~\eqref{eq:bayes} leads to the following unnormalized posterior density:
\begin{align}
    \label{eq:posterior}
    f(\bm{A}, \bm g,\bm p,q\mid\bm{X}) 
    \propto & \prod_{1\leq i<j \leq n} \left[{t \choose X_{ij}}  p_{A_{ij}}^{X_{ij}}\cdot (1-p_{A_{ij}})^{t-X_{ij}}\right]^{\delta_{g_i g_j}}\times\left[
    {t \choose X_{ij}} q^{X_{ij}}(1-q)^{t-X_{ij}}\right]^{1-\delta_{g_i g_j}} \times \mathds{1}_{p_{-1}\leq p_0\leq p_1}(\bm p) \enspace.
\end{align}
In principle, this distribution is enough to reconstruct the network, by marginalizing the parameters out, as described in Eq.~\eqref{eq:bayes:marg}, and calculate the posterior mean from it using Eq.~\eqref{eq:marginal_mean}. 
This posterior density does not correspond to a well-known family of distributions, meaning that inference is not straightforward either.
Hence, we turn to Markov Chain Monte Carlo (MCMC) techniques.

\section{Network Reconstruction Algorithm}
\label{sec:algorithm}

We define a Markov chain over networks $\bm{A}$, partitions $\bm g$ and parameters $\bm{p},q$ whose stationary distribution corresponds to the posterior density defined in Eq.~\eqref{eq:posterior}~\cite{hastings1970monte,robert2004monte_carlo,wang2022effective}. 
We implement a coordinate-wise Metropolis-Hastings (MH) algorithm that cycles through blocks of parameters in a deterministic order to generate random updates~\cite{Johnson_2013,robert2004monte_carlo}.
Sampling proceeds in two steps: a \emph{proposal step}, in which a new state is sampled for one of the variables, and an \emph{acceptance step}, in which the update is accepted with probability:
\begin{equation}
    \alpha = \mathrm{min}\left(1, R \right) \enspace,
\end{equation}
where
\begin{align*}
    R &= \frac{f(\bm{A}', \bm g', \bm p',q' \mid \bm{X})}{f(\bm{A}, \bm g, \bm p, q \mid \bm{X})} \frac{\kappa (\bm{A}', \bm g', \bm p',q \to \bm{A}, \bm g, \bm p, q)}{\kappa(\bm{A}, \bm g, \bm p,q \to \bm{A}', \bm g', \bm p',q')} \enspace,
\end{align*}
and $\kappa(\bm{A}, \bm g, \bm p,q  \to \bm{A}', \bm g', \bm p',q')$ is the probability of proposing a transition from the state $(\bm{A}, \bm g, \bm p,q)$ to a different state $(\bm{A}', \bm g', \bm p',q')$.
The ratio
\begin{equation}
    \frac{\kappa (\bm{A}', \bm g', \bm p',q \to \bm{A}, \bm g, \bm p, q)}{\kappa(\bm{A}, \bm g, \bm p,q \to \bm{A}', \bm g', \bm p',q')} \enspace,
\end{equation}
is a correction term and accounts for asymmetries in the proposal distribution. 
We use a symmetric proposal distribution whenever we can, so that the ratio of transition probabilities is $1$. This approach circumvents the necessity of calculating the correction term for the acceptance probability, reducing $R$ to a ratio of posterior densities only,
\begin{equation}
    R = \frac{f(\bm{A}', \bm g', \bm p',q' \mid \bm{X})}{f(\bm{A}, \bm g, \bm p,q \mid \bm{X})} \enspace.
\end{equation}

\subsection{Sampling algorithm}
\label{subsec:proposals}
For the network $\bm{A}$, we sample proposals by choosing a pair of nodes uniformly at random and switching the edge value to one of the remaining categories with equal probability.
This proposal is symmetric and the ratio of transition probabilities is thus $1$.
To calculate the ratio of posterior densities, we let $\bm{A}$ and $\bm{A}'$ be the current network and the proposed network, respectively, and $(i,j)$ be the index of the edge where these two networks differ.
We then have
\begin{align*}
    R =
    (1 - \delta_{g_i,g_j}) +
    \delta_{g_i,g_j}
    \frac{
        p_{A'_{ij}}^{X_{ij}}\cdot (1-p_{A'_{ij}})^{t-X_{ij}}
    }{
        p_{A_{ij}}^{X_{ij}}\cdot (1-p_{A_{ij}})^{t-X_{ij}}
    },
\end{align*}
meaning that a switch is always accepted when it involves two nodes in different groups, otherwise it is accepted with a probability that depends on the interaction parameters~$\bm p$.

For the intragroup interaction probabilities~$\bm p$, we again use a simple symmetric proposal distribution and perturb the current values with Gaussian noise.
More specifically, if we let $p_{\text{sign}}$ be the parameter we want to update, we define the random variable $v\sim N(0,\sigma_{\rm intra})$, where $\sigma_{\rm intra}$ is the standard deviation, set to a constant throughout the execution of the algorithm.
In our experiments, we fixed $\sigma_{\rm intra}=0.01$, as we found it provided good acceptance rates. 
The proposed new parameter is then $p'_{\text{sign}}=p_{\text{sign}}+v$. 
The proposal is symmetric as the Gaussian distribution is symmetric around zero, so the acceptance probability simplifies to, 
\begin{equation}
    \label{eq:intragroup_ratio}
    R =
    \mathds{1}_{p'_{-1}<p'_0<p'_1}(\bm p')
    \prod_{\substack{
        i<j\\
        g_i= g_j\\
        A_{ij}=\mathrm{sign}
    }}
    \frac{
        {p'}_{\mathrm{sign}}^{X_{ij}} (1-{p'}_\mathrm{sign})^{t-X_{ij}}
    }{
        p_{\mathrm{sign}}^{X_{ij}} (1-p_\mathrm{sign})^{t-X_{ij}}
    }\enspace.
\end{equation}

Note that although the proposal distribution has support over all real numbers, the posterior density at $\bm p'$ is only non-zero when the ordering constraint $0<p_{-1}<p_0<p_1<1$ is satisfied.
The indicator function $\mathds{1}_{p'_{-1}<p'_0<p'_1}(\bm p')$ in Eq.~\eqref{eq:intragroup_ratio} ensures that all proposals that would break this condition will be rejected in the acceptance step. 
The above product is taken over all pairs of nodes $(i,j)$ within the same group that are connected by an edge of the given sign.

Similarly, the intergroup interaction probability $q$ is updated with the same symmetrical proposal distribution, but with standard deviation $\sigma_{\rm inter}=0.1$, and the corresponding ratio is
\begin{equation*}
    R =
    \prod_{\substack{
        1\leq i< j \leq n\\
        g_i\neq g_j
    }}
    \frac{
        (q')^{X_{ij}} (1-q')^{t-X_{ij}}
    }{
        (q)^{X_{ij}} (1-q)^{t-X_{ij}}
    } \enspace.
\end{equation*}

Finally, we define updates to the node partition, $\bm g$. 
Here, we must relinquish symmetrical proposals because this would result in a vast number of low-probability steps.
Instead, we introduce a proposal distribution that leads to more probable transitions~\cite{peixoto2020merge}, by proposing a node to move to groups where it has a higher number of observed interactions. 
We construct the proposed partition $\bm g'$ by picking a node $i$ uniformly at random, and proposing to move it from its current group, $g_i$, to another one, $g'_i=k$.
We then let $\bm g'$ be the proposed partition in which node $i$ is now in group $g'_i$, and use the following proposal distribution
 \begin{align*}
     \kappa(\bm g \to \bm g') =
     &\begin{cases}
         \frac{1+\sum^n_{j\in N(g'_i)}  X_{ij}}{\sum_j^n X_{ij}+\gamma+1} \enspace, & \text{ if }1\leq k\leq \gamma\\
         \frac{1}{\sum_j^n X_{ij}+\gamma+1} \enspace, &\text{ if }k=\gamma+1\\
           0 \enspace, &\text{ else} \enspace,
     \end{cases}
 \end{align*}
where $\gamma$ is the number of non-empty groups in $\bm g$ and $N(k)$ denotes the nodes of group $k$.
This update 
creates groups with non-zero probability, and groups can be removed by moving the last node of a group to a different group. 
In this case we re-index groups from 1 to $\gamma$ to maintain a contiguous labeling.
The corresponding ratio is then
\begin{align*}
    R = 
    \prod_{i<j} \frac{
            \delta_{g'_i g'_j}     p_{A_{ij}}^{X_{ij}} (1\!-\!p_{A_{ij}})^{t\!-\!X_{ij}}
            \!+\!(1\!-\!\delta_{g'_i g'_j}) q^{X_{ij}}(1\!-\!q)^{t\!-\!X_{ij}}
    }{
             \delta_{g_i g_j}    p_{A_{ij}}^{X_{ij}} (1\!-\!p_{A_{ij}})^{t\!-\!X_{ij}}
            \!+\!(1\!-\!\delta_{g_i g_j}) q^{X_{ij}}(1\!-\!q)^{t\!-\!X_{ij}}
    }
    \times
    \frac{
        \kappa(\bm g'\to\bm g)
    }{
        \kappa(\bm g\to \bm g') 
    } \enspace.
\end{align*}

\subsection{Algorithmic implementation}
\label{sec:alg_ref}
To summarize the algorithm, we initialize all parameters and ensure that the constraint $p_{-1}<p_0<p_1$ is respected.
We then cycle through the parameters and update them sequentially, in the following order: (i)~the network $\bm{A}$, (ii)~each of the intra-group connection probabilities $p_{-1},p_0,p_1$,  (iii)~the inter-group connection probabilities $q$, and (iv)~the partition $\bm g$.
For each parameter, we generate a new proposal step and accept it with probability $\alpha=\min(1,R)$ where the specific ratio $R$ is given by the above equations.
The resulting algorithm provably generates a chain whose equilibrium distribution is given by Eq.~\eqref{eq:posterior} (see Appendix~\ref{appendix:convergence} for details).


\subsection{Augmenting the algorithm to multiple observations}
\label{subsec:multi_observation}
Before we present the results, we note that our algorithm can be easily modified to process multiple sets of observations on the same node set---for example, observations among a fixed set of individuals gathered over different time periods or in different contexts (work versus school, weekends versus weekdays, etc.).
In this context, we now have $r$ observation matrices $\bm{X}^{(1)},...\bm{X}^{(r)}$, and observation periods of different lengths $t^{(1)},...,t^{(r)}$, collectively denoted as $\{\bm{X}^{(s)}\}_{s=1}^r$, and $\{t^{(s)}\}_{s=1}^r$. 
We model the latent social network $\bm{A}$ as constant, but every other parameter is allowed to change between observation periods.
We form the likelihood of this updated model by treating each observation period as conditionally independent, i.e.,
\begin{equation}
    \mathbb P \left(\{\bm{X}^{(s)}\}_{s=1}^r \mid \bm{A},\{\bm g^{(s)}, \bm p^{(s)},q^{(s)},t^{(s)}\}_{s=1}^r\right) 
    =\prod_{s=1}^r \mathbb P\left(\bm{X}^{(s)}|\bm{A},\ \bm g^{(s)}, \bm p^{(s)},q^{(s)}, t^{(s)}\right) \enspace,
    \label{eq:likelihood_multiobservation}
\end{equation}
where we have made explicit the dependency on the length of the observation period, $t^{(s)}$.
Using the same line of reasoning as before, we obtain the posterior density: 
\begin{equation}
    f \Big(\bm{A},\{\bm g^{(s)}, \bm p^{(s)},q^{(s)} \}_{s=1}^r \mid  \{\bm{X}^{(s)}, t^{(s)}\}_{s=1}^r \Big) 
    \propto  \prod_{s=1}^r  \, \mathbb P \, \Big(\bm{X}^{(s)}|\bm{A},\ \bm g^{(s)}, \bm p^{(s)},q^{(s)}, t^{(s)}\Big)  
     \times   \mathds 1_{p_{-1}\leq p_0\leq p_1}\left(\bm p^{(s)}\right) \enspace.
    \label{eq:posterior_multiple_obs}
\end{equation}

Sampling the posterior in Eq.~\eqref{eq:posterior_multiple_obs} can be achieved by applying the same MCMC algorithm as before, just requiring that we cycle through the multiple observation periods.
The acceptance probabilities are straightforwardly calculated as products of the ratios $R$ derived in Section~\ref{subsec:proposals} above.

\section{Results}
\label{sec:results}
We evaluate our signed network reconstruction method, first on synthetic data where the true network is known to confirm that our method works as intended.
Then, we apply it to real-world data, in which the true network is unknown, and demonstrate its practical utility through indirect validation techniques. 

\subsection{Synthetic data experiments}
We benchmark our method against competing approaches and assess its robustness to different prior distributions. The scripts to replicate all results are available online.~\footnote{\url{github.com/ferenczid/signed_net_inf}.}

\subsubsection{Performance Evaluation}

To assess the performance of the algorithm, we generate synthetic data, by generating signed networks, interaction probabilities and network partitions.
We first generate an Erd\H{o}s-R\'enyi random graph with $n=64$ nodes and edge density $0.4$.
Then, we assign positive or negative signs to each edge with equal probability. We initialize the partitions by placing all nodes into either a single group (if target internal edge fraction is larger than $0.5$) or distinct individual groups (if target internal edge fraction is less than $0.5$). We then vary the internal edge fraction by making micro-adjustments---swapping one node to a new group at a time---and only keeping the swaps that move the total number of internal edges closer to our exact target, repeating this exact process until the target fraction is met.

The interaction probabilities, $p^+,p^0,p^-$ are drawn from uniform distributions $\text{Unif}(0.8,0.9)$, $\text{Unif}(0.2,0.3)$ and $\text{Unif}(0,0.1)$ respectively.
With these parameters in hand, we then generate a single interaction matrix $\bm{X}$ drawn from the likelihood in Eq.~\eqref{eq:likelihood}.
Our goal is to assess the performance of the reconstruction algorithm given the observation~$\bm{X}$.

We focus our assessment on the signed network $\bm{A}$, as it is the main target of inference for the algorithm.
We evaluate the reconstruction performance by computing the area under the ROC curve (AUC) for each class using a one-versus-rest approach, treating one class in turn as the positive class against all others~\cite{fawcett2006introduction}.
Often in multi-class classification problems, the performance is reported as a single number by combining the AUC scores across all classes in some manner~\cite{hand2001simple, domingos2000well}.
However, here we are interested in how well the algorithm can distinguish each edge type from the rest, and so we examine these AUC scores separately.

To add context to our method's performance, we also implement four simple baseline algorithms: one based on frequency thresholding, one based on the configuration model~\cite{bollobas1980probabilistic}, one based on ordinal regression~\cite{winship1984regression}, and one based on the stochastic block model~\cite{peixoto_mesoscale}.

For frequency thresholding, we use two thresholds on the raw interaction counts: pairs above the upper threshold are classified as positive, those below the lower threshold as negative, and those between as neutral. The AUC for each class is computed by varying the corresponding threshold.

The configuration model (CM) baseline uses relational information, as recorded in $\mathbf{X}$, to make a classification.
In the CM baseline, we first calculate an expected interaction matrix
\begin{equation}
    \langle \bm{X} \rangle=\frac{1}{2}\frac{\bm{d} \bm{d}^\top }{\bm{d}^\top \bm{d}} \enspace,
\end{equation}
where $\bm{d}$ is a length $n$ vector with entries $d_i=\sum_{j=1}^n \bm X_{ij}$ corresponding to the total number of interactions recorded for individual $i$.
The matrix $\langle \bm{X} \rangle$ can be viewed as encoding the expected number of interactions in a weighted configuration model for the network formed by interpreting the number of observed interactions between each pair of individuals as a weighted edge~\cite{newman2018network}.
We then calculate how the observed matrix $\bm{X}$ deviates from this expectation as $\hat{\bm{A}}^{\mathrm{(CM)}}=\bm{X}-\langle \bm{X} \rangle$, and interpret the resulting matrix as an estimate of a signed network,\footnote{We could also normalize the matrix, but it would not affect the classification outcome.} classifying $(i,j)$ as a positive edge if $\hat{\bm{A}}_{ij}>0$ and as a negative edge otherwise, while scoring neutral edges based on their absolute proximity to the expected baseline ($-\lvert \hat{\bm{A}}_{ij} \rvert$).

The ordered-probit baseline leverages interaction count data, but does not incorporate any network information.
In the ordinal-regression baseline, we randomly reveal a fraction $\omega\in[0,1]$ of the data to fit an ordered-probit model that predicts the sign $\hat{A}^{(\rm OR)}_{ij}\in\{-1,0,1\}$ from the counts $\bm{X}$.
More specifically, we let $A^{*}_{ij}$ be a normally distributed latent score with mean $\beta\, \bm X_{ij}$ and unit variance.
Then, we estimate
\begin{equation}
    \hat{A}^{(\rm OR)}_{ij} =
    \left\{
\begin{array}{r@{\qquad}l@{\,}c@{\,}r}
-1, &        & A^{*}_{ij} & \leq \tau_1,\\
0,  & \tau_1 < & A^{*}_{ij} & \leq \tau_2,\\
1,  & \tau_2 < & A^{*}_{ij} &,
\end{array}
\right.
\end{equation}
where $\beta,\tau_1,$ and $\tau_2$ are estimated by maximum likelihood.
This model can then be used to predict the sign of the interaction of the remaining pairs of nodes.
In our tests, we use a very generous $\omega=0.5$, exposing half of the actual signs to train the model.

Our last baseline is based on the edge-valued stochastic block model~(SBM) described in \cite{peixoto_mesoscale}, where we use the observation matrices as the edge-valued adjacency matrix and fit a stochastic block model. Since this is an unsigned method, we turn the inferred mesoscale structure into a sign prediction method by mapping intra-block homophily to positive ties. Specifically, to generate the continuous scores required for AUC evaluation, we rank node pairs using a two-tier hierarchy: pairs are sorted primarily by their block-membership status (prioritizing intra-block over inter-block for positive ties, and vice versa for negative ties), and secondarily by their observed interaction frequency $X_{ij}$.
\begin{figure*}
    \centering
    \includegraphics[width=1\textwidth]{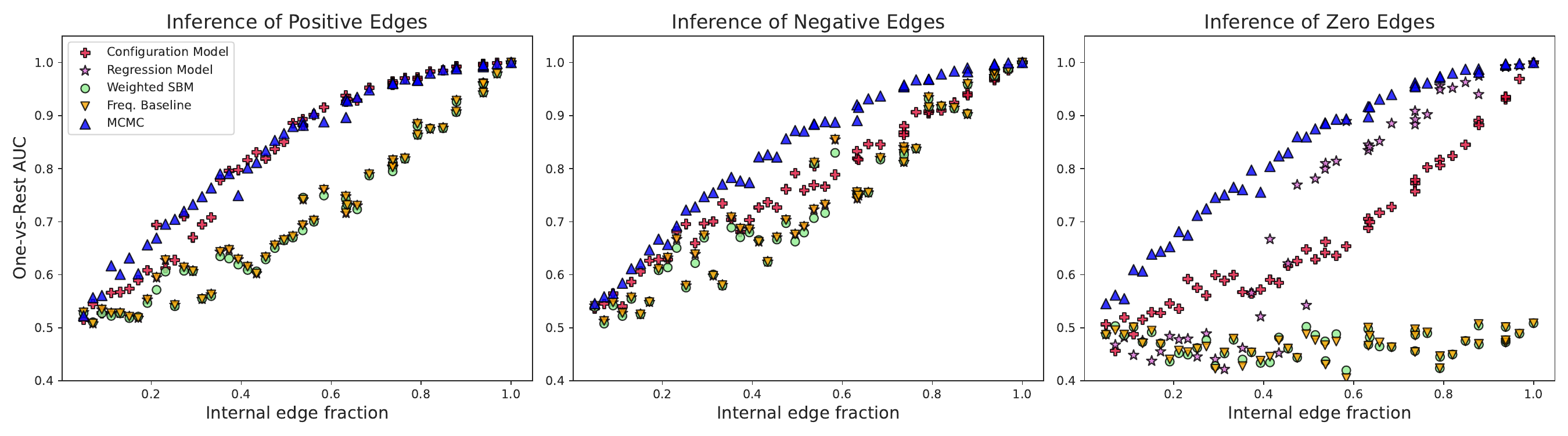}
    \caption{
        Reconstruction accuracy for our MCMC algorithm and four baselines on (\textbf{left}) positive, (\textbf{middle}) negative, and (\textbf{right}) neutral edge classification.
        We generate synthetic signed networks with 64 nodes and assign them to groups.
        We vary the number of groups in the network partition and record the fraction of edges within groups (internal edge fraction).
        For each network and partition, we generate an observation matrix $\bm{X}$.
        AUC scores are shown for one-versus-rest classification as a function of the internal edge fraction.
    }
    \label{fig:conf_mod}
\end{figure*}

Figure~\ref{fig:conf_mod} displays our results on this synthetic experiment.
We find that as we vary the fraction of edges within groups, both the negative and positive edges become easier to distinguish from others across all methods.
Our method (MCMC) unsurprisingly performs best, as it encodes the generative assumptions used to produce the synthetic data.

The frequency thresholding baseline performs similarly to the other methods for positive and negative edge inference but struggles to distinguish neutral edges from the other two types. The configuration model attains near-MCMC performance for positive edges at a much lower computational cost, though its performance is lower for negative and neutral edges.

The ordinal-regression baseline, despite access to half the true edge labels during training, achieves intermediate performance across all three categories. The SBM baseline performs comparably for positive edges but is among the weaker performers for negative edges and similarly struggles to distinguish neutral edges. Neutral edge inference is the most discriminating benchmark: MCMC maintains a clear advantage throughout, while the other baselines show substantially reduced performance, particularly frequency thresholding and SBM.

\subsubsection{Prior robustness}\label{sec:prior_experiments}
To assess the robustness of our results to the choice of prior, we do a parameter sweep and check the relative error between the reconstructed and ground-truth edges when using different priors. Figure~\ref{fig:prior_par_sweep} compares the error of our uniform prior against three cases: 
\begin{enumerate}
    \item the true prior used to generate the data;
    \item a ``sparse'' prior that assumes 5\% positive edges, 5\% negative edges and the rest absent, and
    \item a ``rare negatives'' prior in which only 1\% of edges are assumed to be negative and 20\% positive. 
\end{enumerate}

The plots show that our choice of uniform prior performs similarly or better in most cases. It only performs poorly when the prior forbids the formation of certain edge types, i.e., when the network has only one or two edge types, along the edges and corners of the simplex.

\begin{figure*}
    \centering
    \includegraphics[width=1\linewidth]{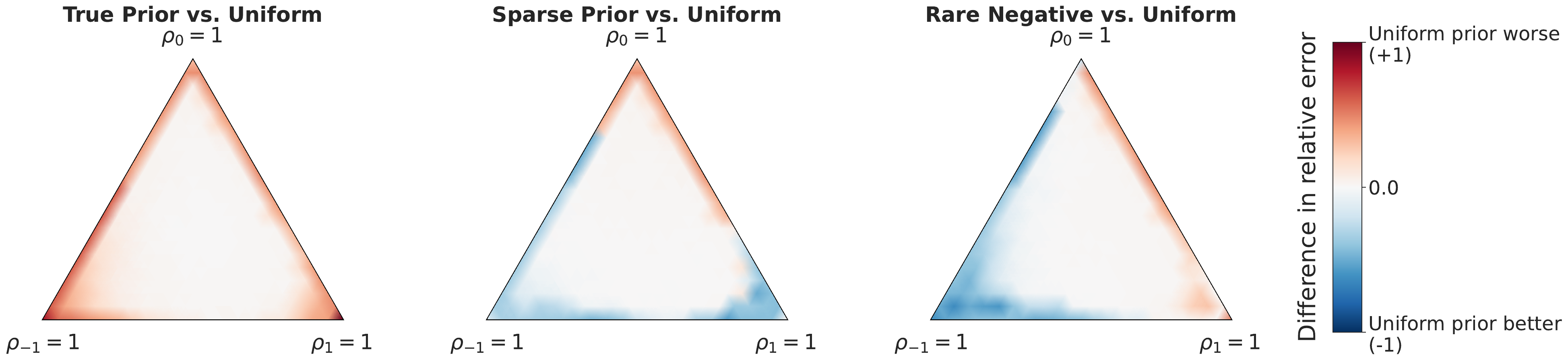}
    \caption{
    Robustness of the uniform prior for reconstructing signed edges. Each point of the equilateral triangle corresponds to a set of network parameters $(\rho_{-1}, \rho_0, \rho_1)$ used to generate interaction data. Color indicates the relative reconstruction error of the uniform prior $(1/3, 1/3, 1/3)$ compared against other choices of prior. \emph{Left:} True prior used to generate the data. \emph{Center:} Sparse prior $(0.05, 0.90, 0.05)$. \emph{Right:} Rare Negatives prior $(0.01, 0.79, 0.20)$.\label{fig:prior_par_sweep}}
\end{figure*}


\subsection{Real-world data}
Next, we demonstrate our method on an empirical human contact dataset to infer social ties from observed interactions. 

\subsubsection{Dataset}
In 2013, Mastrandrea et al.~\cite{Mastrandrea_Fournet_Barrat_2015} conducted a multi-level analysis of a French high school social network as part of the SocioPattern collaboration~\cite{cattuto2010dynamics}.
They collected multiple datasets over the span of five days (December 2-6, 2013), including contact data obtained with wearable sensors and contact diaries.

To gather the contact data, each student received a device that recorded with whom that student interacted throughout the week.
Data collection was carried out in 20-second time windows; that is, if two individuals faced each other and were within a 1--1.5 meter range, their devices registered the interaction for that time window~\cite{Mastrandrea_Fournet_Barrat_2015}.
The resulting data are encoded as a binary interaction tensor, where each slice represents an interaction matrix for a given time window.

Additional metadata was gathered, including but not limited to: self-reported friendships obtained via questionnaires and the school class membership of each student. Note that the class membership is distinct from the unobserved partition $\bm g$ that we infer with our model to indicate opportunity to interact freely at a given time.

We restrict our analysis to interactions during lunch breaks, when students are more likely to act according to their social preferences, thereby aligning the data more closely with the assumptions of our model and supporting its construct validity.
The timing of these breaks is not explicitly provided in the dataset, so we estimate them by tracking within-class interactions; when the number of cross-class interactions exceeds a certain threshold, we record the timestamp as the start of the break, and conversely, record the end of a break when that number drops below the threshold, as shown in Fig.~\ref{fig:real_results}(a).

\begin{figure*}[ht]
\includegraphics[width=1\columnwidth]{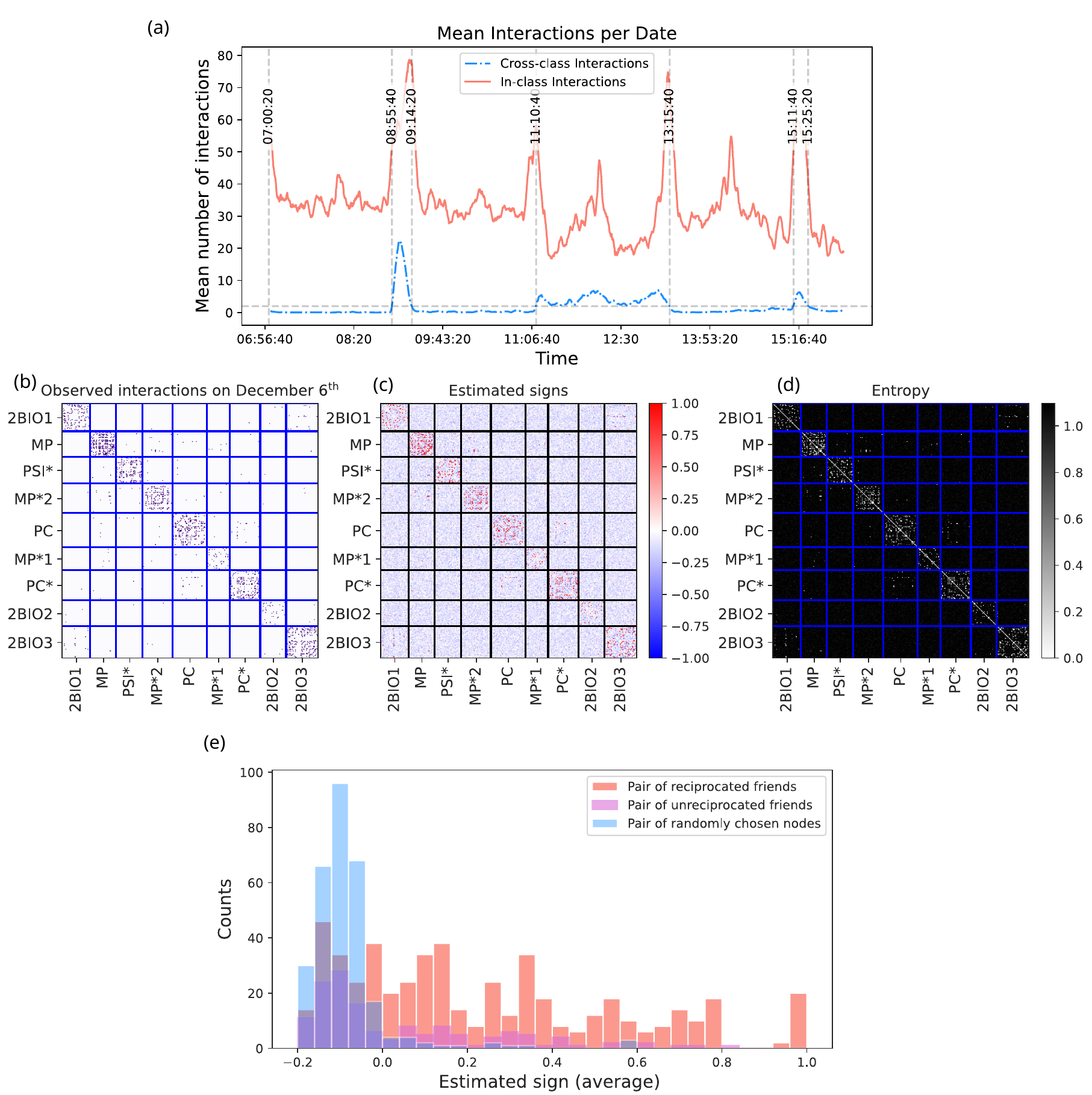}
      \caption{
      Signed network reconstruction from contact data among high school students. 
      (\textbf{a}) Mean number of interactions of students during the day averaged over the five days of the study. 
      In-class interactions are shown in orange, while cross-class interactions are shown in blue. 
      The vertical lines represent the start and end of the breaks, which are estimated from the rise of cross-class interactions.
      (\textbf{b}) Dichotomized matrix, in which pairs of students who interacted at all during the breaks are shown with a dark square.
      Classes are delimited with blue lines and denoted by specialization, e.g., ``PC'' for the ``Physics-Chemistry'' class.
      (\textbf{c}) Inference results, in which $\hat{A}_{ij}$ is the posterior mean of the sampled interaction signs, computed from 64 samples.
      (\textbf{d}) Entropy of the estimated distribution over signs for each pair of students.
      Pairs with high entropy, corresponding to low certainty, are shown in black.
      The pairs with the lowest entropy consist of students in the same class.
      (\textbf{e}) Using questionnaire data, we identify 262 reciprocated and 144 unreciprocated friendships.
      This histogram shows the average of the sample means that we infer across each of the five networks, one for each day of the study.
      A control group, comprised of 300 pairs of nodes selected uniformly at random, is shown in blue.}
  \label{fig:real_results}
\end{figure*}

We identified breaks at 8:55 and 15:10, both lasting 15 minutes, while lunch ran from 11:10 to 13:15.
After filtering the interactions that fell into these specific time windows, we construct our observation matrices $\bm{X}^{(1)},\ldots,\bm{X}^{(r)}$. 
To determine the optimal time window length for our observation period, we must balance two competing considerations: Windows that are too long may overlook crucial information as multiple social groups form and dissolve undetected, while windows that are too short lack sufficient signal for the fitted model to describe the data accurately. We settled on splitting the breaks into 4-minute periods, a data-driven choice designed to capture stable periods of social interaction without smoothing over rapid group reorganization (see Appendix \ref{app:temp_coex} for an empirical evaluation confirming the structural consistency of this window).

\subsubsection{Network inference}
Within each day, we assume the latent network remains constant and apply the multi-observation model described in Sec.~\ref{subsec:multi_observation} to the observation matrices obtained by breaking down the lunch breaks in slices of four minutes. To robustly infer the posterior distribution of the network partitions and interaction rates, we run an ensemble of 64 independent MCMC chains simultaneously. Comprehensive diagnostics confirming the convergence and stability of this ensemble are provided in Appendix \ref{app:mcmc_convergence}.

Figure~\ref{fig:real_results}~(b-d) gives an overview of our results for December 6\textsuperscript{th}, 2013, the last day of the study. 
Columns and rows in each matrix in panels~(b-d) correspond to students ordered by the class they attend throughout the day, e.g., ``PC'' for ``Physics--Chemistry'' or ``MP'' for ``Mathematics--Physics.'' 
Panel~(b) is a dichotomized version of the interaction data to show which students interacted at least once during the breaks that day. We see that most students interacted with their classmates even during breaks. 
Panel~(c) displays the inferred signed network, where each entry $\hat{A}_{ij}$ is the posterior mean (Eq.~\eqref{eq:marginal_mean}) of the sampled interaction signs between students $i$ and $j$, computed from the 64 posterior samples.
We see that positive edges are concentrated within classes, indicating that students tend to have positive relationships with their classmates.
Importantly, however, we do not infer strongly negative edges between classes even if the number of recorded interactions is close to 0 for every such pair of students.
Instead, the model acknowledges that data are lacking and maintains high posterior uncertainty about the students' relationship.
We illustrate this in panel~(d), where we plot the entropy of the marginal posterior probability $h(A_{ij})=f(A_{ij} \mid \{\bm{X}^{(s)}\}_{s=1}^r, \{t^{(s)}\}_{s=1}^r )$, i.e.,
\begin{equation}
    H_{ij} =\,\,\, -\!\!\!\!\!\!\!\!\sum_{A_{ij}\in \\ \{-1, 0, 1 \}} \!\!\!\!\!\! h(A_{ij}) \log h(A_{ij}) \enspace,
\end{equation}
for every pair of students.
The entropy is maximal ($-\log 3 \approx 1.1$ nats) outside of the diagonal, where students are in different classes and interactions are lacking, while it is low within classes, where we have more interactions.

Since we do not have access to the true social network, we validate our results using questionnaire data collected as part of the study~\cite{Mastrandrea_Fournet_Barrat_2015}.
We process the questionnaires to identify reciprocated friendships (where both students nominated each other as friends) and unreciprocated friendships (where only one student nominated the other).
We also sample 300 pairs of students uniformly at random to form a control group.
Our hypothesis is that reciprocated friendships would be associated with more positive edges on average, followed by unreciprocated friendships, and that the random pairs would be mostly indifferent or negative.
This is precisely what we see in panel~(e) of Fig.~\ref{fig:real_results}.
However, we also find that a minority of edges connecting friends have a sample mean similar to that of our control. 
This indicates that those friendship edges were not active during the observation period, i.e., they were not spending time in the same social group during the lunch breaks of the week.

Figure~\ref{fig:scatter_argmax} plots the posterior probability for each edge type against the observed interaction count $X_{ij}$. A natural decision boundary emerges near $X_{ij} \approx 6$: above it, the model assigns high probability to positive edges for nearly all pairs. Below this threshold, unlabeled pairs receive high zero-edge probability, while reciprocated and unreciprocated friend pairs retain elevated $P(A_{ij}=+1)$ even with sparse interaction evidence, showing that sign inference captures relational structure beyond interaction frequency alone. Table~\ref{tab:silent_friends_full} in the Appendix lists specific low-frequency friend pairs ($X_{ij} \leq 5$) inferred as positive with high confidence.

\begin{figure}[ht]
    \centering
    \includegraphics[width=1\linewidth]{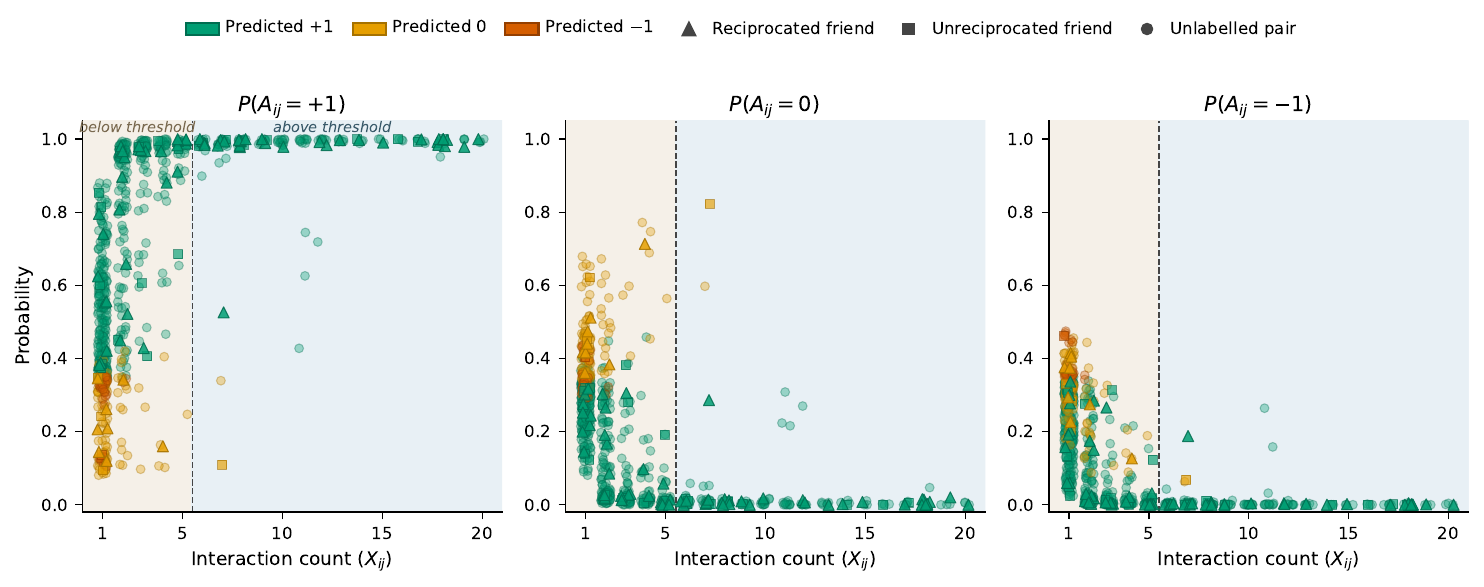}
    \caption{Posterior probability for positive (left), zero (middle), and negative (right) edges as a function of observed interaction count $X_{ij}$. Points are colored by argmax prediction and shaped by self-reported friendships (reciprocated friend, unreciprocated friend, unlabeled pair). A dashed vertical line marks the interaction-count threshold ($X_{ij} \approx 6$). Friend pairs show elevated $P(A_{ij}=+1)$ even below the threshold.}
    \label{fig:scatter_argmax}
\end{figure}

\subsubsection{Posterior predictive check}
\label{sec:posterior_predictive}
Although our model and method can give us an estimate of the latent social network for \emph{any} set of observation matrices $\bm{X}^{(1)},...\bm{X}^{(r)}$, there is no guarantee that the fit will be any good.
This could arise if, for example, the model is missing a critical component of the data-generating process.
Hence, we also complete our analysis with a posterior predictive analysis to validate the model's fit~\cite{gelman2013bayesian,young2021bayesian}.
We essentially compare the original data with new data simulated using the estimated network and parameters.
If they are dissimilar, then the model is a poor fit and needs to be revised.

\begin{figure}
    \includegraphics[width=1\columnwidth]{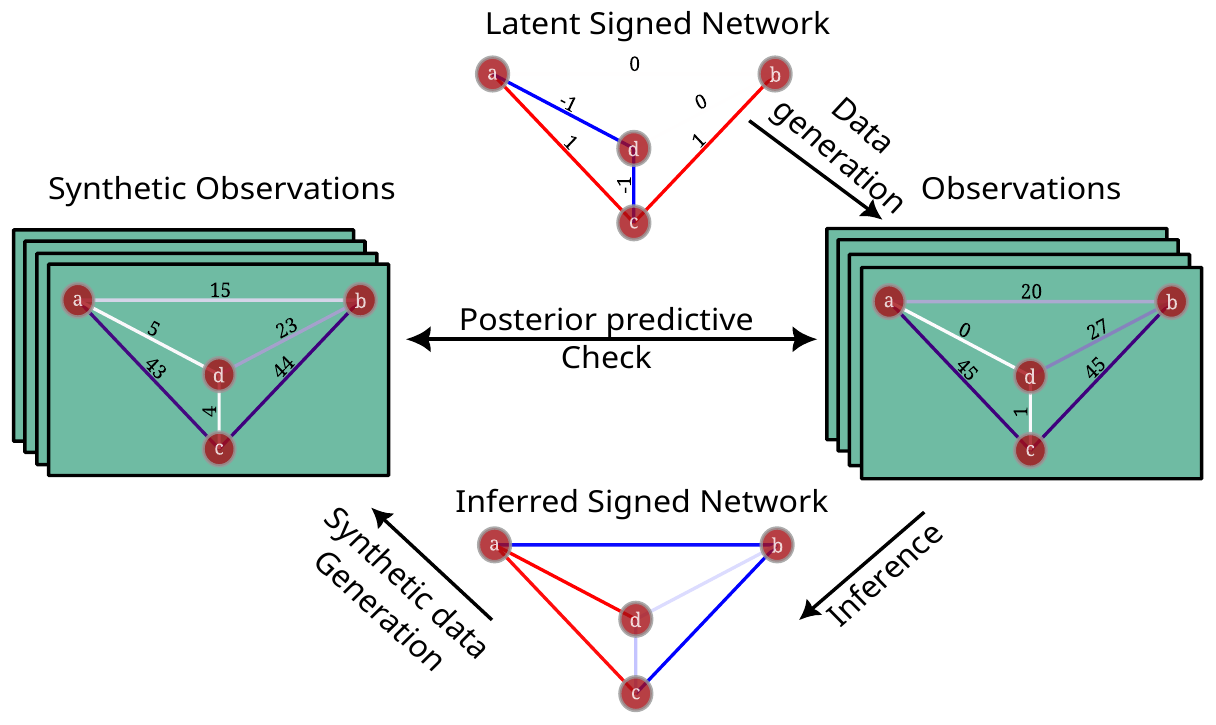}
    \caption{
        Posterior predictive analysis.
        The latent signed network (\textbf{top}) is never observed, but it determines the observations (\textbf{right}).
        Our method can be used to estimate a network using these observations alone (\textbf{bottom}).
        In posterior predictive analysis, we then generate a new set of synthetic observations using the generative model (\textbf{left}), and compare them to the actual observations (middle arrow).
        If they are dissimilar, then the model is a poor fit and needs to be revised.
    }
    \label{fig:post_pred_sol}
\end{figure}

Figure~\ref{fig:post_pred_sol} illustrates our posterior predictive check workflow. 
We let $\{\bm{X}^{(s)}\}_{s=1}^r$ denote the observation matrices for any given day of the study, where each $s$ corresponds to a 4-minute observation period.
The network $\bm{A}$ is assumed to be constant throughout the day. 
We obtain $64$ samples of the adjacency matrix from the posterior, each denoted by $\bm{A}_\ell$ for $\ell\in\{1,...,64\}$.
Since the groups and parameters can potentially differ for each observation period, we let $\bm p^{(s)}_{\ell}$ and $q^{(s)}_{\ell}$ denote the $\ell$\textsuperscript{th} sample of the parameters drawn for period $s$, and likewise for the partition, $\bm g^{(s)}_{\ell}$.

For each 4-minute period $s$, we then generate a new observation matrix $\widetilde{\bm{X}}^{(s)}_{\ell}$ using the likelihood in Eq.~\eqref{eq:likelihood_multiobservation}, thereby pairing the samples of the posterior distribution with synthetic observations.

Next, we compare these synthetic observations to the observed matrices, using two different statistics, the log-likelihood ratio discrepancy~\cite{gelman1996posterior}, and the total number of interactions per node.
For an observation matrix $\bm{X}$ and a set of parameters $(\bm{A},\bm g,\bm p,q)$, the log-likelihood discrepancy captures the distance between the expectation and the actual observations as 
\begin{equation*}
    D(\bm{X}; \bm{A}, \bm g, \bm p,q) = \sum_{1 \leq i< j \leq n} X_{ij} \log \frac{X_{ij}}{\mathbb{E}[ X_{ij}]} \enspace,
\end{equation*}
where $\mathbb{E}[X_{ij}]$ is the expected observation matrix when the unknown parameters equal $(\bm{A},\bm g, \bm p,q)$. 
A model cannot be rejected if the distribution of the discrepancy $D(\widetilde{\bm{X}}; \bm{A}, \bm g,\bm p,q)$ of synthetic data is similar to that of the observed data, $D(\bm{X}; \bm{A},\bm g,\bm p,q)$.
And this test can be summarized with a Bayesian $p$-value giving the probability that
\begin{equation}
    D(\bm{X}; \bm{A},\bm g,\bm p,q ) < D(\widetilde{\bm{X}}; \bm{A},\bm g,\bm p, q) \enspace.
\end{equation}

Using our paired sample described above, we evaluate this discrepancy for the synthetic observations and the actual observations for each time period.
Figure~\ref{fig:post_pred_comb} shows these discrepancies for one time period and $64$ posterior samples, with each dot corresponding to one posterior sample.
The distribution of the discrepancy is similar for the synthetic and observed data, meaning that we have not obtained any evidence suggesting a lack of fit.
The corresponding $p$-value of $0.25$ confirms this observation.

\begin{figure}
    \includegraphics[width=1\columnwidth]{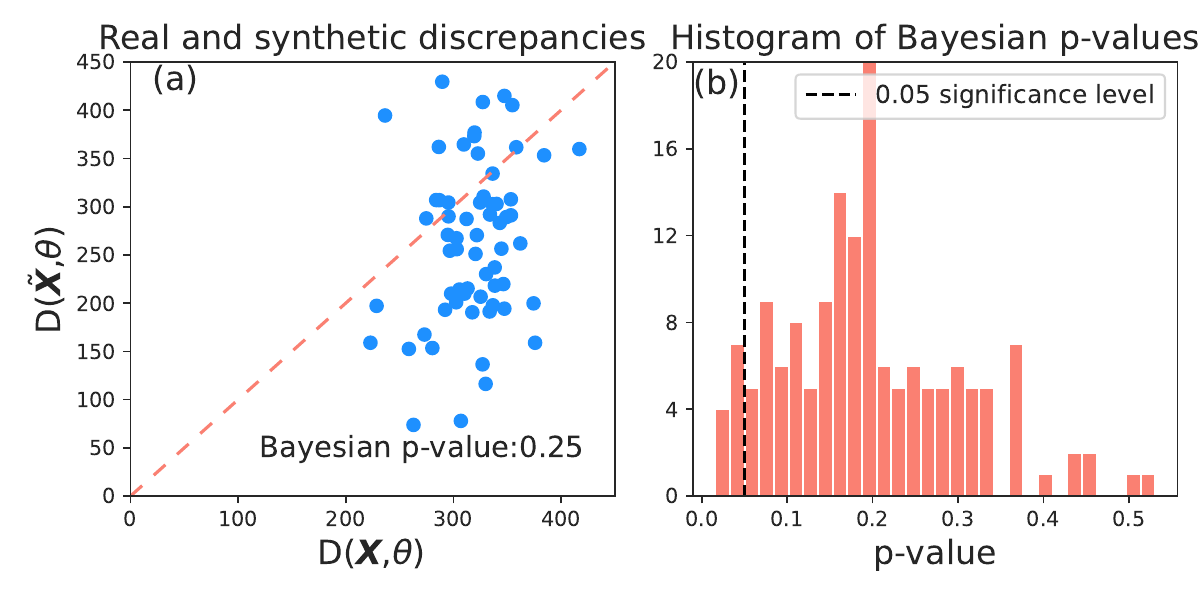}

	\includegraphics[width=1\linewidth]{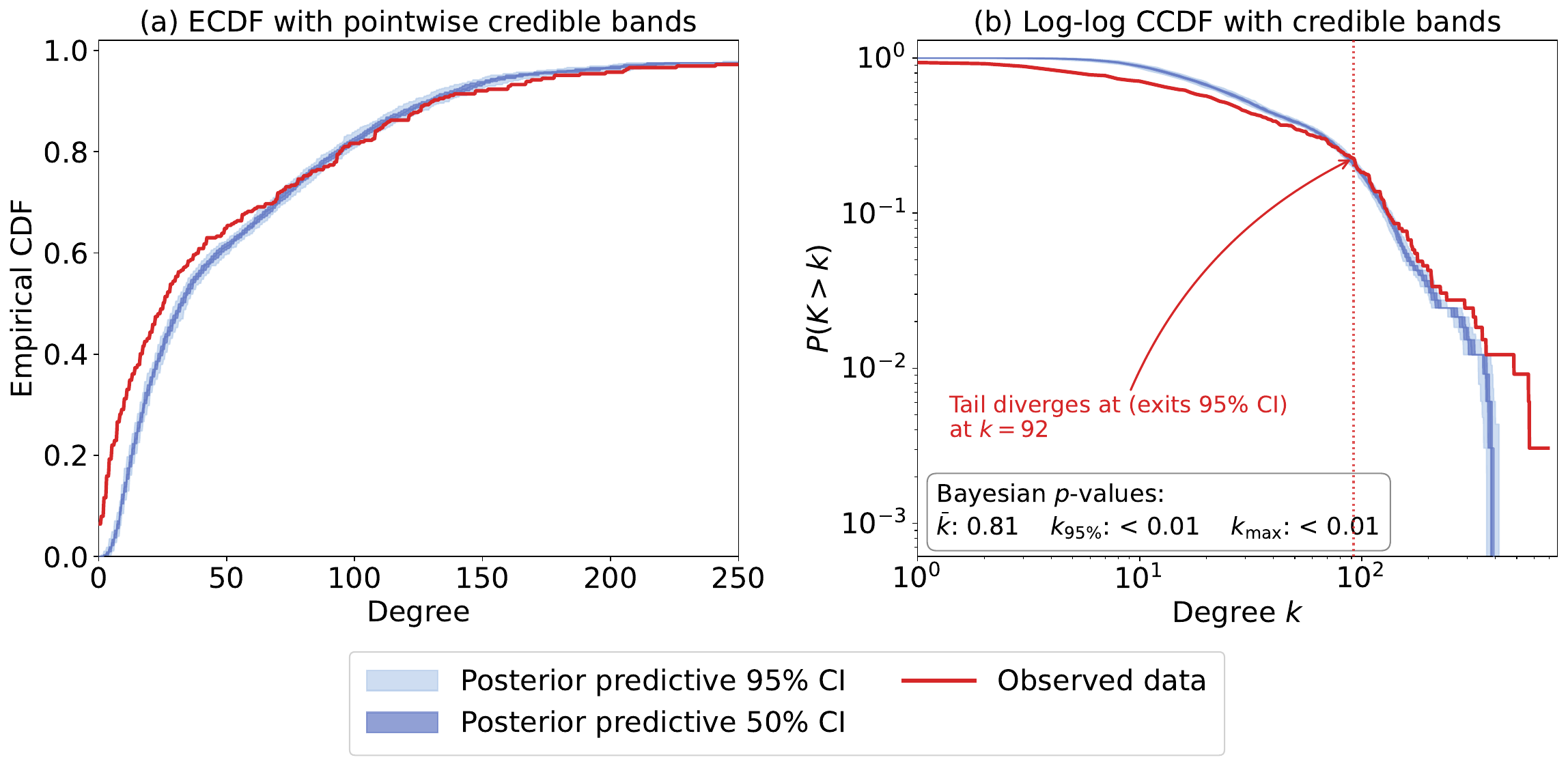}
	\caption{\textbf{Top:}
		(a) Example of the posterior predictive for one observation period.
		The diagonal line indicates equal values of the discrepancies.
		This test yields a Bayesian $p$-value of $0.25$, corresponding to the fraction of nodes above the dotted line.
		(b) Histogram of Bayesian p-values for the high school dataset.
		The dotted line represents the (arbitrary) $0.05$ significance level.
		\textbf{Bottom:} Posterior predictive check for the degree distribution showing the ECDF (a) and log-log CCDF (b) with 50\% and 95\% credible bands. The mean degree is well captured ($p = 0.81$), but the 95th-percentile and maximum degree fall outside the credible interval ($p < 0.01$), indicating that the model underproduces high-degree hub nodes.}
	\label{fig:post_pred_comb}
\end{figure}

This result replicates across nearly all observation periods, as shown in Fig.~\ref{fig:post_pred_comb}.
Of the 156 time periods we fit, we find that only 11 yield a $p$-value under the typical  $0.05$ significance level. 
Note that we should not necessarily expect a flat distribution of $p$-values, as the datasets share a common parameter, the latent network $\bm{A}$, and are thus not independent.
The values below $0.05$ (an arbitrary threshold) do not necessarily indicate poor performance in our method. Instead, these values could suggest that something occurred during these specific observation periods, such as the rapid formation and dissolution of groups, or some other external factor, leading to poorer fit.

Since the log-likelihood discrepancy is dominated by high-frequency pairs, we complement it with a degree-distribution check. Figure~\ref{fig:post_pred_comb} shows the posterior predictive ECDF and log-log CCDF of node degree with 50\% and 95\% credible bands. The mean degree is well reproduced ($p = 0.81$), but the 95th-percentile and maximum degree fall well outside the posterior predictive interval ($p < 0.01$), revealing that the model underproduces high-degree hub nodes. This points to a degree-heterogeneous model extension as a natural direction for future work.

\section{Conclusion}
Here we have developed a probabilistic generative model and a corresponding Bayesian inference algorithm for inferring signed networks from contact pattern data.
Our method provides a way to reconstruct signed networks where direct observations of relationships are unavailable, by leveraging indirect observations in the form of interaction counts.
Specifically, we have focused on distinguishing between the absence of interactions due to negative relationships and the absence of interactions due to indifference or lack of opportunity to interact.
We validated the effectiveness of our method on synthetic data, and evaluated its applicability and robustness on a real-world dataset of face-to-face interactions~\cite{Mastrandrea_Fournet_Barrat_2015}. 
Our comparison with metadata on self-reported friendships provided evidence that our inferred networks aligned well with the students' own perceptions of their social ties.
Finally, we demonstrated the goodness-of-fit of our generative model through posterior predictive checks.

Our work opens several avenues for future research. 
We emphasize that, although we have focused on a specialized case study, our model is highly customizable and could be easily adapted to a much broader set of problems.
For example, our generative model could naturally be extended to bipartite graphs, which would be useful for analyzing a wide range of datasets, from ecological networks to recommendation systems.

Additionally, we could refine the data model itself, tailoring our method to better capture different collection processes related to social network analysis, e.g., using survey data or observing online interactions in social media. Furthermore, the conditional independence assumption could be relaxed by incorporating a copula structure~\cite{shang_percolation}, such as the Farlie--Gumbel--Morgenstern copula, to model dependence between observations more accurately.
Another promising direction would be to examine unobserved or partially observed nodes in the network, by incorporating more sophisticated network priors.

\section*{Acknowledgments}

This work was supported in part by the National Science Foundation award SES-241973 (JGY), by the National Institutes of Health 1P20 GM125498-01 Centers of Biomedical Research Excellence Award (JGY), the Dutch Research Council (NWO) Talent Programme ENW-Vidi 2021 under grant number VI.Vidi.213.163 (LP) and (DF).
This research was made possible, in part, using the Data Science Research Infrastructure (DSRI) hosted at Maastricht University. This work used the Dutch national e-infrastructure with the support of the SURF Cooperative using grant no. EINF-15302.
\bibliography{bibliography}

\null\vspace{2\baselineskip}
\appendix
\section{Proof of Convergence for the Markov-Chain}
\label{appendix:convergence}

In the main text, we have constructed the transition probabilities with the Metropolis-Hastings method and, as a result, the conditions for detailed balance are satisfied, and the stationary distribution of the resulting Markov chain exists~\cite{robert2004monte_carlo,hastings1970monte}.

Now, we need to verify the uniqueness condition of the stationary distribution. 
Reference~\cite{robert2004monte_carlo} demonstrates that the ergodicity of the Markov chain guarantees this condition.
Hence, we prove the following theorem:

\begin{theorem}
    The Markov chain described in Sec.~\ref{sec:algorithm} is ergodic.
\end{theorem}

\begin{proof}

Let $(\bm{A}^x,p^x_{-1},p^x_0,p^x_1,q^x,g^x)$, and $(\bm{A}^y,p^y_{-1},p^y_0,p^y_1,q^y,g^y)$ be two arbitrary states of the Markov chain with positive probabilities. 
We need to demonstrate that there exists a set of moves that, with positive probability, transition the state from state $x$ to state $y$. 

For the interaction probabilities, this condition is easily satisfied due to the properties of the normal and the uniform distribution.

For the network and partitions, let $\mu_{\bm A_{x}^y}$, and $\mu_{g_{x}^y}$ be an enumeration of steps that get $\bm A_x$ and $g_x$ to $\bm A_y$ and $g_y$ respectively.

Let $\mu$ be the concatenation of $\mu_{\bm A_{x}^y}$, and $\mu_{g_{x}^y}$. Then the following holds:

\begin{gather*}
    \mathbb P(\text{ getting from } \bm{A}^x,g^x \text{ to } \bm{A}^y,g^y)\\
    \geq \mathbb P(\text{ we make the transitions listed in } \mu)>0,
\end{gather*}
which concludes our proof.
\end{proof}

\section{Structural Consistency of Partition Model}\label{app:temp_coex}
  To ensure our 4-minute observation windows capture stable social configurations rather than transient noise, we empirically evaluate the network's temporal persistence. Rather than relying on heuristic global clustering metrics or imposing arbitrary confidence thresholds, we work directly with the full posterior co-assignment matrix to evaluate structural consistency. For a given observation period $s$, we define the $n \times n$ matrix $C^{(s)}$ where each entry represents the posterior probability that nodes $i$ and $j$ are co-assigned to the same group: 
	$$C^{(s)}_{ij} = \mathbb{P}(g^{(s)}_i = g^{(s)}_j \mid \bm{X}^{(s)})$$
	This formulation yields a probabilistic summary of the inferred group structure that retains full posterior uncertainty, rather than prematurely collapsing it into a binary decision. To measure temporal stability between intervals separated by a temporal gap of $\tau$ steps (i.e., $s$ and $s+\tau$), we compute the normalized Frobenius inner product (cosine similarity) between their respective co-assignment matrices. To ensure that trivial self-assignments ($i=j$) do not artificially inflate the baseline similarity, we restrict the summation strictly to the upper triangle of the matrices, thereby isolating true interpersonal relationships:
	\begin{equation*}
		\text{sim}(s, s+\tau) = \frac{\sum_{i < j} C^{(s)}_{ij} C^{(s+\tau)}_{ij}}{\sqrt{\sum_{i < j} \left(C^{(s)}_{ij}\right)^2} \sqrt{\sum_{i < j} \left(C^{(s+\tau)}_{ij}\right)^2}}\end{equation*}
	Because this metric evaluates the two matrices as vectors, it naturally weights all pairs by their co-assignment probability at both intervals simultaneously. Pairs that are confidently co-assigned at both $s$ and $s+\tau$ contribute significantly to the similarity score, while inactive nodes with diffuse, near-uniform probabilities contribute negligible signal automatically, without the need to handle them as a special case. To contextualize whether the group structure is genuinely persistent relative to the timescale of the data, we evaluate this metric across multiple temporal gaps $\tau \in \{1, 2, 3, 4, 5\}$. As shown in Fig.~\ref{fig:temporal_class_persistence}, the inferred partitions are highly persistent over short timescales, with structural similarity demonstrating a clear, monotonic decay as the temporal gap $\tau$ grows.
\begin{figure}[htbp]
	\centering
	\includegraphics[width=0.8\columnwidth]{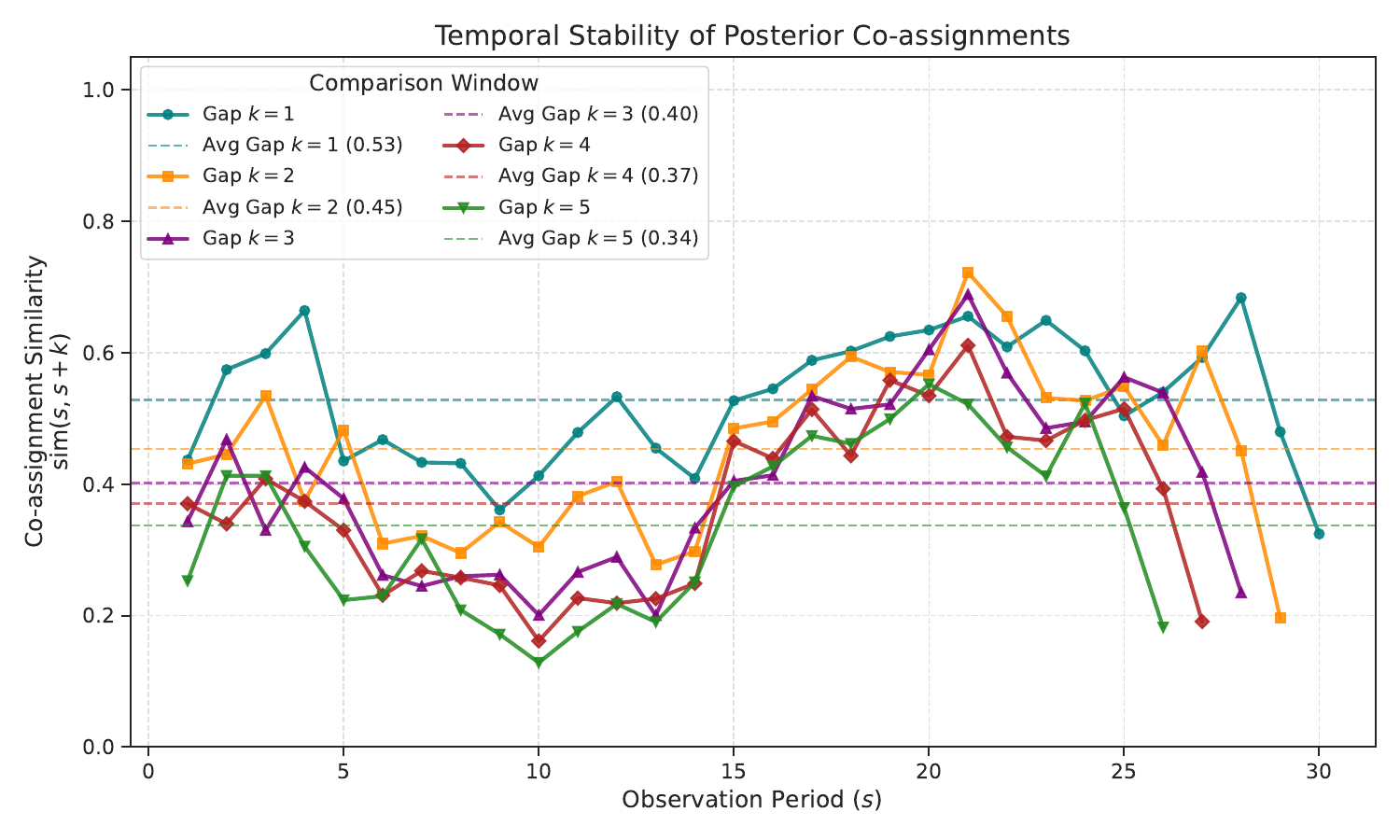}
	\caption{Temporal persistence of confident opportunity groups. The conditional posterior persistence of co-assignments is high on average, implying that confidently inferred social groups remain highly stable throughout the break. The expected persistence drops heavily towards the end of the observation period, capturing the abrupt dissolution of these social groups as students are forced to return to class and reorganize according to their official schedules. }
	\label{fig:temporal_class_persistence}
\end{figure}
\section{MCMC Convergence Diagnostics and Ensemble Stability} \label{app:mcmc_convergence}

Evaluating MCMC convergence in discrete network partition models requires distinguishing between within-chain stationarity and between-chain mixing. Because the space of possible partitions is combinatorial, chains naturally lock into different, closely situated local optima (modes) \cite{peixoto2014efficient}. Consequently, standard convergence metrics that penalize between-chain variance (such as the Gelman-Rubin $\hat{R}$) must be interpreted alongside metrics that evaluate ensemble-level stability and precision \cite{Vehtari}.

\subsection{Initialization and Stationarity}
To ensure a robust characterization of the posterior distribution, we deployed an ensemble of 64 independent MCMC chains. We enforced a strict separation between two distinct phases of the sampling process:
\begin{enumerate}
	\item \textbf{Initialization:} The initial phase after which discrete partition inference begins and chains locate high-probability regions.
	\item \textbf{Statistical Sampling:} A secondary burn-in period (discarding the first 150 iterations) to ensure the interaction rate parameters ($p_{+1}, p_{0}, p_{-1}$) have reached local stationarity before sample collection.
\end{enumerate}

We confirm the success of this two-phase approach visually in Figure \ref{fig:combined_convergence}. The linear scale illustrates the diverse initialization of the 64 chains, while the logarithmic scale demonstrates that all chains achieved stationarity well before the sampling threshold.
\subsection{Convergence Metrics and Precision}
\begin{figure*}[htbp]
	\centering
	\begin{subfigure}[b]{0.49\textwidth}
		\centering
		\includegraphics[width=\textwidth]{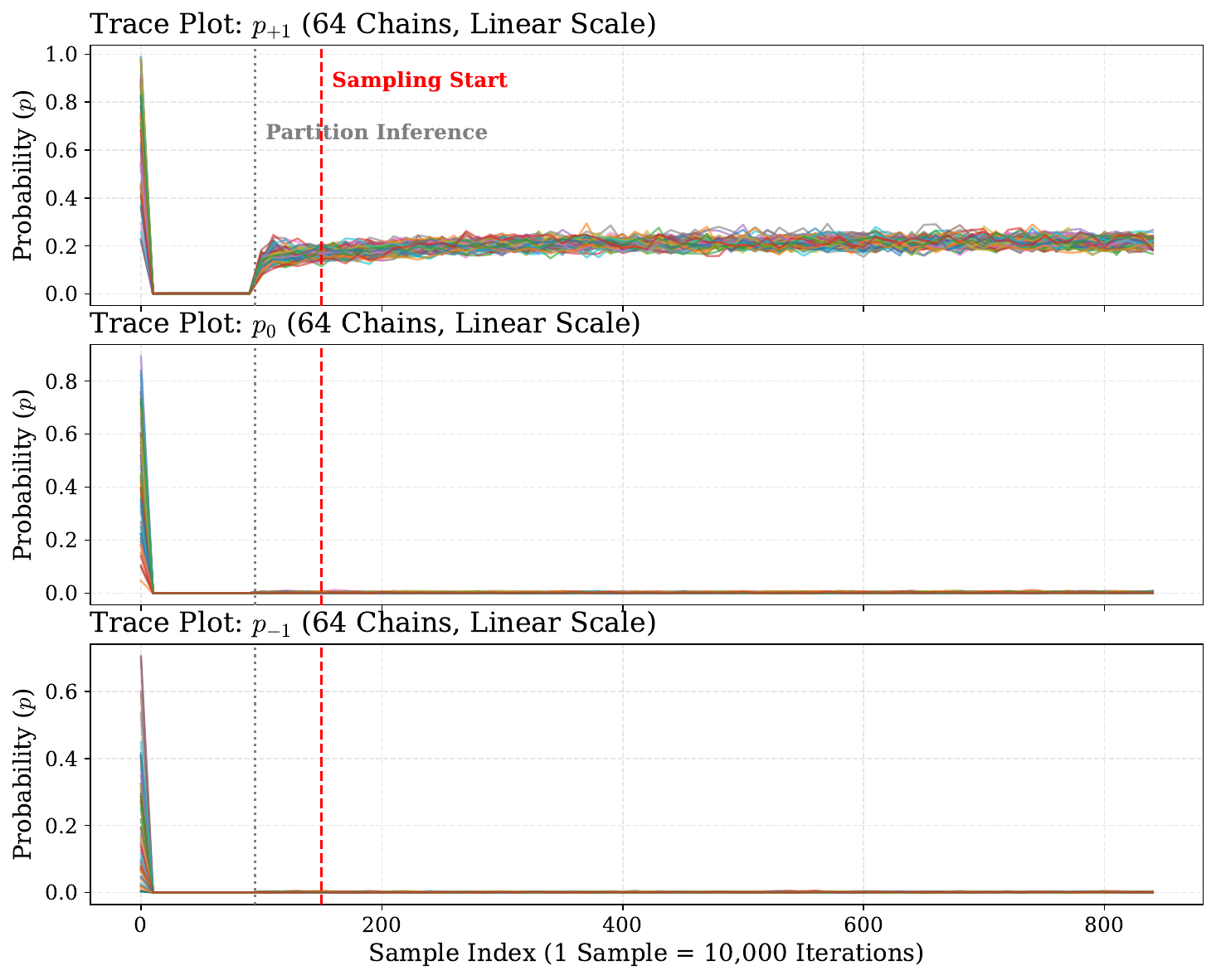}
		\caption{Linear scale: Highlights initialization diversity.}
		\label{fig:trace_linear}
	\end{subfigure}
	\hfill
	\begin{subfigure}[b]{0.49\textwidth}
		\centering
		\includegraphics[width=\textwidth]{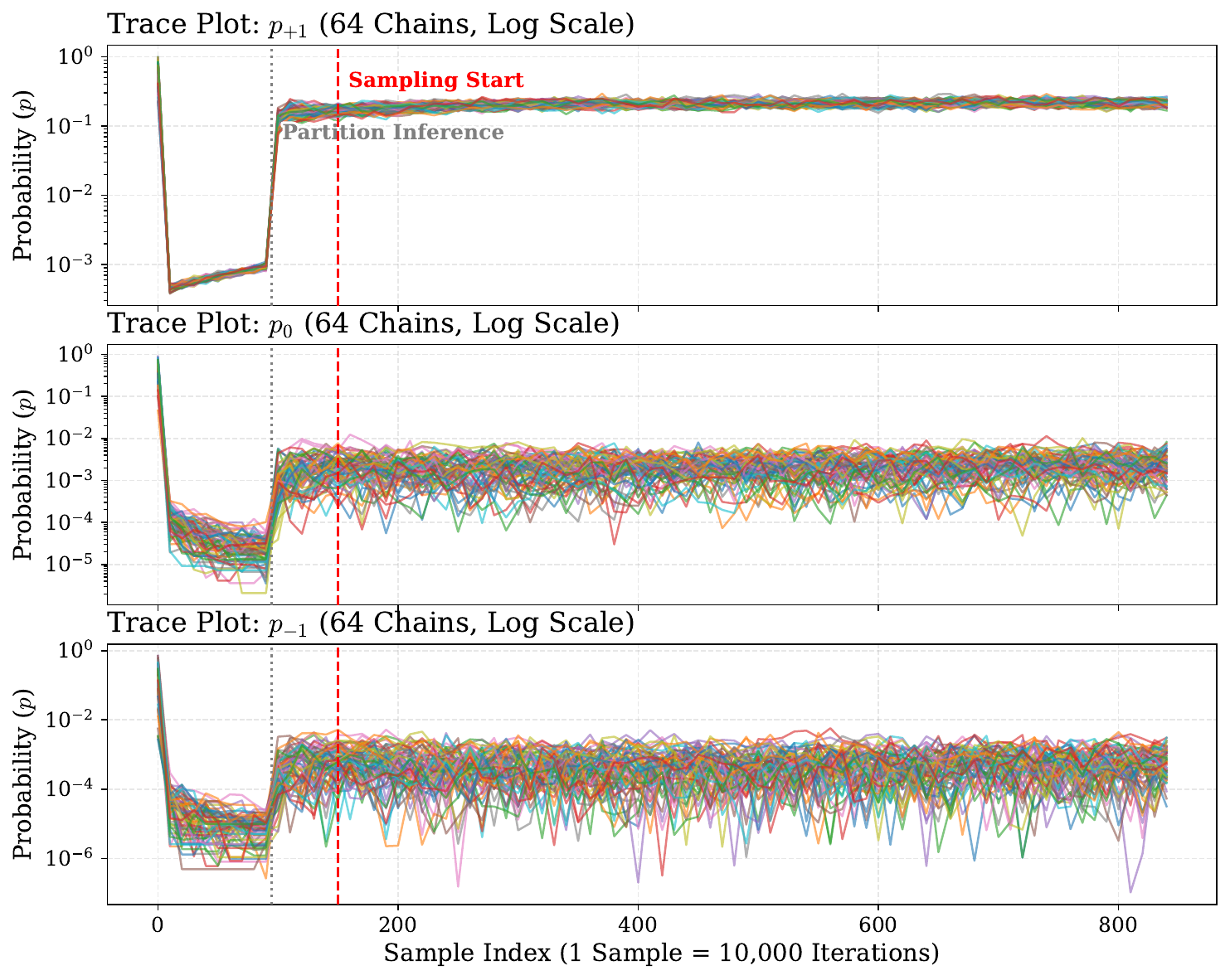}
		\caption{Logarithmic scale: Preserves sampling dynamics.}
		\label{fig:trace_log}
	\end{subfigure}
	\caption{MCMC trace plots for the interaction rate parameters $p_{+1}, p_{0}, p_{-1}$. The linear scale (left) demonstrates the diverse initialization of the 64 chains, while the logarithmic scale (right) illustrates the stable stationarity of the ensemble across several orders of magnitude.}
	\label{fig:combined_convergence}
\end{figure*}
As established above, the rugged, multimodal nature of the network partition space causes standard between-chain metrics like the Gelman-Rubin statistic ($\hat{R}$) to remain elevated ($\approx 1.2$), as individual chains stabilize in spatially distinct local optima. However, as shown in Table \ref{tab:mcmc_metrics}, our within-chain diagnostics confirm the validity and precision of our results despite this spatial separation:
\begin{itemize}
	\item \textbf{Precision:} The Monte Carlo Standard Error (MCSE) is negligible (e.g., $< 0.002$ for $p_{+1}$), proving that the mean estimates are stationary and the sampling error is minimal.
	\item \textbf{Sufficiency:} The pooled Tail Effective Sample Size ($ESS_{tail}$) for all parameters comfortably exceeds 500, surpassing the standard threshold of 400 required to robustly define 95\% Highest Density Intervals (HDIs).
\end{itemize}

\begin{table}[htbp]
	\centering
	\caption{MCMC Convergence Diagnostics for the 64-chain ensemble (post burn-in). The elevated $\hat{R}$ reflects expected structural multimodality, while low MCSE and high $ESS_{tail}$ confirm precise, sufficient sampling.}
	\label{tab:mcmc_metrics}
	\begin{tabular}{@{}lccc@{}}
		\toprule
		\textbf{Parameter} & \textbf{$\hat{R}$} & \textbf{MCSE} & \textbf{Tail ESS} \\ \midrule
		$p_{+1}$ & 1.230 & 0.0016 & 536 \\
		$p_{0}$  & 1.306 & 0.0001 & 482 \\
		$p_{-1}$ & 1.146 & 0.00004 & 562 \\ \bottomrule
	\end{tabular}
\end{table}

\subsection{Ensemble Stability via Split-Half Comparison}
While the previous metrics confirm the quality of individual chains, the primary goal of our ensemble design is to comprehensively map the entire multimodal network space. Therefore, the critical convergence question is whether the global ensemble average itself is stable and reproducible.

To demonstrate this, we conducted a split-half stability test by randomly dividing the 64 independent chains into two distinct groups of 32 chains. Figure \ref{fig:split_half_histogram} shows the overlapping density histograms for these two independent macro-ensembles. The near-perfect overlap demonstrates that the 64-chain ensemble reproducibly maps the exact same posterior distribution. This confirms that the inferred parameters are not artifacts of specific, sticky chain trajectories, but are highly stable, globally representative properties of the model.

\begin{figure}[htbp]
	\centering

	\includegraphics[width=0.9\textwidth]{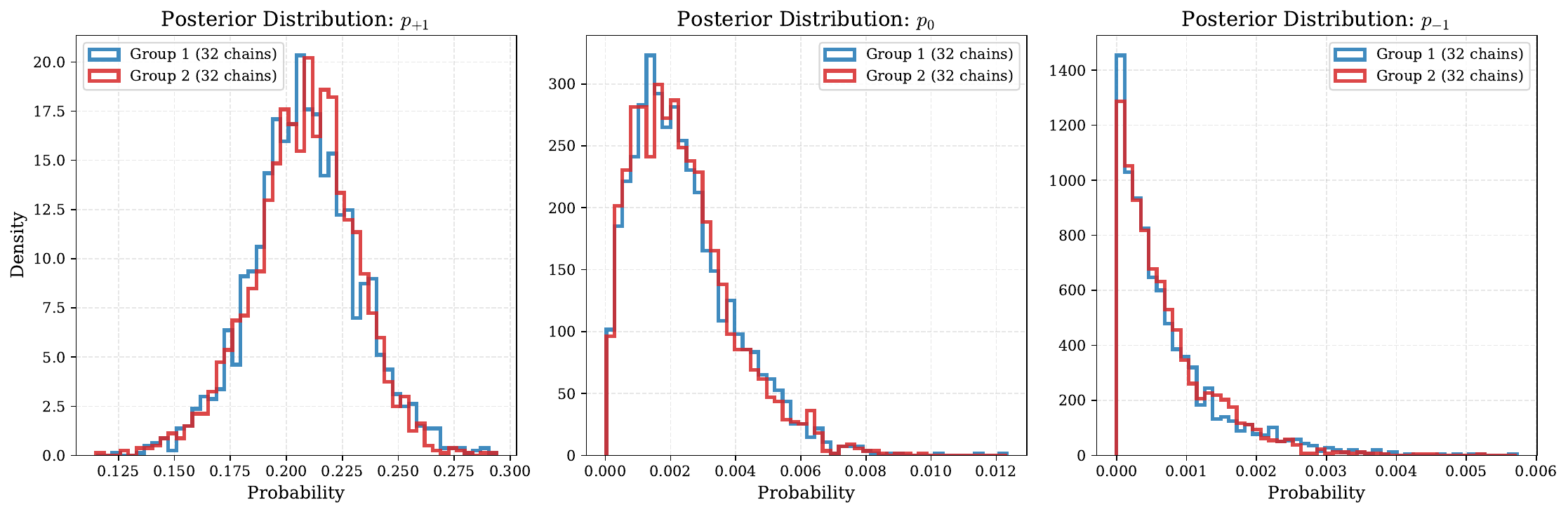}
	\caption{Split-half ensemble stability test. The 64 MCMC chains were randomly divided into two independent groups of 32. The overlapping posterior density histograms confirm that both groups independently map the same multimodal posterior distribution, proving the ensemble is highly stable and reproducible.}
	\label{fig:split_half_histogram}
\end{figure}
	\section{Low Contact Friends}
A strength of our modeling approach is the recovery of social structure from sparse interaction data. Table \ref{tab:silent_friends_full} highlights these 'silent ties': pairs of students who reported a mutually reciprocated friendship in the survey, but for whom the physical sensors recorded minimal interaction ($X_{ij} \leq 5$) during the observation period. Despite the sparsity of physical contact, the model correctly assigns a high posterior probability to a positive tie ($A_{ij} = +1$) by leveraging the global community structure and the interaction patterns of the surrounding network.
\begin{table}[htbp]
	\centering
	\caption{Complete list of high certainty ($\ge0.5$) inferred positive ties for sparse interaction pairs ($X \leq 5$). All pairs are mutually reciprocated friends in the survey.}
	\label{tab:silent_friends_full}
	\begin{tabular}{@{}ccc@{}}
		\toprule
		\textbf{Student IDs} & \textbf{Contacts ($X_{ij}$)} & \textbf{$P(A_{ij} = +1)$} \\ \midrule
		3 \& 407   & 5 & 0.9879 \\
		70 \& 240  & 5 & 0.9952 \\
		70 \& 425  & 2 & 0.8970 \\
		79 \& 335  & 1 & 0.6248 \\
		80 \& 285  & 1 & 0.7955 \\
		101 \& 240 & 5 & 0.9110 \\
		117 \& 196 & 2 & 0.8075 \\
		120 \& 285 & 1 & 0.7406 \\
		125 \& 325 & 5 & 0.9973 \\
		125 \& 624 & 5 & 0.9830 \\
		132 \& 425 & 2 & 0.6587 \\
		184 \& 441 & 5 & 0.9993 \\
		201 \& 642 & 4 & 0.8810 \\
		211 \& 845 & 4 & 0.9672 \\
		245 \& 691 & 2 & 0.9636 \\
		245 \& 869 & 3 & 0.9673 \\
		272 \& 587 & 3 & 0.9719 \\
		325 \& 624 & 2 & 0.9806 \\
		325 \& 959 & 5 & 0.9996 \\
		366 \& 974 & 1 & 0.5534 \\
		407 \& 429 & 3 & 0.9702 \\
		440 \& 920 & 1 & 0.5561 \\
		494 \& 883 & 3 & 0.9805 \\
		502 \& 691 & 2 & 0.9510 \\
		642 \& 691 & 2 & 0.9672 \\
		798 \& 960 & 2 & 0.5218 \\ \bottomrule
	\end{tabular}
\end{table}

\end{document}